\documentclass[11pt,a4paper]{article}

\usepackage{fullpage}
\usepackage[utf8x]{inputenc}
\usepackage{amsmath, amsthm}
\usepackage{wrapfig,graphicx,amssymb,textcomp,array,amsmath}
\usepackage{algpseudocode} 
\usepackage{enumerate}
\usepackage{enumitem}
\usepackage{multirow}
\usepackage{tabularx}
\algtext*{EndWhile}
\algtext*{EndIf}
\usepackage{algorithm}
\usepackage{color}
\usepackage[titletoc,title]{appendix}

\newcommand{\etal}{{et al.}}

\newcommand{\dist}[2]{|#1#2|}

\newcommand{\changed}[1]{{\color{black} #1}}

\title{Faster Algorithms for some Optimization Problems\\ on Collinear Points
}

\author{Ahmad Biniaz\thanks{Cheriton School of Computer Science, University of Waterloo, 
		\newline\indent\indent  ahmad.biniaz@gmail.com, imunro@uwaterloo.ca}
	\and Prosenjit Bose\thanks{School of Computer Science, Carleton University, \{jit, anil, michiel\}@scs.carleton.ca}  
	\and Paz Carmi\thanks{Department of Computer Science, Ben-Gurion University of the Negev, carmip@cs.bgu.ac.il} 
	\and  Anil Maheshwari\footnotemark[2]
	\and  J. Ian Munro\footnotemark[1]
	\and  Michiel Smid\footnotemark[2]}

\date{\today}
\newtheorem{lemma}{Lemma}

\newtheorem{theorem}{Theorem}
\newtheorem{observation}{Observation}
\newtheorem*{problem*}{Problem}

\begin{document}
\maketitle
\begin{abstract}
We propose faster algorithms for the following three optimization problems on $n$ collinear points, i.e., points in dimension one. The first two problems are known to be NP-hard in higher dimensions.

\begin{enumerate}
	\item {\em Maximizing total area of disjoint disks}: In this problem the goal is to maximize the total area of nonoverlapping disks centered at the points. Acharyya, De, and Nandy (2017) presented an $O(n^2)$-time algorithm for this problem. We present an optimal $\Theta(n)$-time algorithm.   
	\item {\em Minimizing sum of the radii of client-server coverage}: The $n$ points are partitioned into two sets, namely clients and servers. The goal is to minimize the sum of the radii of disks centered at servers such that every client is in some disk, i.e., in the coverage range of some server. Lev-Tov and Peleg (2005) presented an $O(n^3)$-time algorithm for this problem. We present an $O(n^2)$-time algorithm, thereby improving the running time by a factor of $\Theta(n)$. 
	\item {\em Minimizing total area of point-interval coverage}: The $n$ input points belong to an interval $I$. The goal is to find a set of \changed{$n$} disks of minimum total area, covering $I$, such that every disk contains at least one input point. We present an algorithm that solves this problem in $O(n^2)$ time.
\end{enumerate}
\end{abstract}

\section{Introduction}
\label{introduction-section}
Range assignment is a well-studied class of geometric optimization problems that arises in wireless network design, and has a rich literature. The task is to assign transmission ranges to a set of given base station antennas such that the resulting network satisfies a given property. The antennas are usually represented by points in the plane. The coverage region of an antenna is usually represented by a disk whose center is the antenna and whose radius is the transmission range assigned to that antenna. In this model, a range assignment problem can be interpreted as the following problem. Given a set of points in the plane, we must choose a radius for each point, so that the disks with these radii satisfy a given property. 

Let $P=\{p_1,\dots,p_n\}$ be a set of $n$ points in the $d$-dimensional Euclidean space. A {\em range assignment} for $P$ is an assignment of a transmission range $r_i\geqslant 0$ (radius) to each point $p_i\in P$.
The cost of a range assignment, representing the power consumption of the network, is defined as
$C=\sum_i{r_i^\alpha}$ for some constant $\alpha\geqslant 1$.
We study the following three range assignment problems on a set of points on a straight-line ($1$-dimensional Euclidean space).

\begin{description}
	\item[Problem 1] Given a set of collinear points, maximize the
	total area of nonoverlapping disks centered at these points. The nonoverlapping constraint requires $r_i+r_{i+1}$ to be no larger than the Euclidean distance between $p_i$ and $p_{i+1}$, for every $i\in\{1,\dots,n-1\}$.
	\vspace{4pt}
	\item[Problem 2] Given a set of collinear points that is partitioned into
	two sets, namely clients and servers, the goal is to minimize the sum of the radii of disks centered at the servers such that every client is in some disk, i.e., every client is covered by at least one server.
	\vspace{4pt}
	\item[Problem 3] \changed{Given a set $\{p_1,\dots, p_n\}$ of $n$ points on an interval, find a set $D_1,\dots,D_n$ of $n$ disks covering the entire interval such that the total area of disks is minimized and for every $i$ the disk $D_i$ contains the point $p_i$.}
\end{description}

In Problem 1 we want to maximize $\sum r_i^2$, in Problem 2 we want to minimize $\sum r_i$, and in Problem 3 we want to minimize $\sum r_i^2$. These three problems are solvable in polynomial time in 1-dimension. Both Problem 1 and Problem 2 are NP-hard in dimension $d$, for every $d \geqslant 2$, and both have a PTAS \cite{Acharyya2017b, Alt2006, Bilo2005}. 

Acharyya~\etal~\cite{Acharyya2017b} showed that Problem 1 can be solved in $O(n^2)$ time. Eppstein~\cite{Eppstein2016} proved that an alternate version of this problem, where the goal is to maximize the sum of the radii, can be solved in $O(n^{2-1/d})$ time for any constant dimension $d$.
Bil\`{o}  \etal~\cite{Bilo2005} showed that Problem 2 is solvable in polynomial time by reducing it to an integer linear program with a totally unimodular constraint matrix. Lev-Tov and Peleg \cite{Lev-Tov2005} presented an $O(n^3)$-time algorithm for this problem. They also presented a linear-time 4-approximation algorithm. Alt~\etal~\cite{Alt2006} improved the ratio of this linear-time algorithm to 3. They also presented an $O(n\log n)$-time 2-approximation algorithm for Problem 2. Chambers~\etal~\cite{Chambers2011} studied a variant of Problem 3\textemdash on collinear points\textemdash where the disks centered at input points; they showed that the best solution with two disks gives a $5/4$-approximation.
Carmi~\etal~\cite{Carmi2006} studied a similar version of the problem for points in the plane. 

\subsection{Our Contributions}
In this paper we study Problems 1-3. In Section~\ref{Problem-1-section}, we present an algorithm that solves Problem 1 in linear time, provided that the points are given in sorted order along the line. This improves the previous best running time by a factor of $\Theta(n)$. In Section~\ref{Problem-2-section}, we present an algorithm that solves Problem 2 in $O(n^2)$ time; this also improves the previous best running time by a factor of $\Theta(n)$. In Section~\ref{Problem-3-section}, first we present a simple $O(n^3)$ algorithm for Problem 3. Then with a more involved proof, we show how to improve the running time to $O(n^2)$.

\section{Problem 1: Disjoint Disks with Maximum Area}
\label{Problem-1-section}
In this section we study Problem 1: Let $P=\{p_1,\dots,p_n\}$ be a set of $n\geqslant 3$ points on a straight-line $\ell$ that are given in sorted order. We want to assign to every $p_i\in P$ a radius $r_i$ such that the disks with the given radii do not overlap and their total area, or equivalently $\sum r_i^2$, is as large as possible. Acharyya \etal~\cite{Acharyya2017a} showed how to obtain such an assignment in $O(n^2)$ time. We show how to obtain such an assignment in linear time. 

\begin{theorem}
	Given $n$ collinear points in sorted order in the plane, in $\Theta(n)$ time, we can find a set of nonoverlapping disks centered at these points that maximizes the total area of the disks.
\end{theorem}

With a suitable rotation we assume that $\ell$ is horizontal. Moreover, we assume that $p_1,\dots,p_n$ is the sequence of points of $P$ in increasing order of their $x$-coordinates. We refer to a set of nonoverlapping disks centered at points of $P$ as a {\em feasible solution}. We refer to the disks in a feasible solution $S$ that are centered at $p_1,\dots, p_n$ as $D_1, \dots, D_n$, respectively. Also, we denote the radius of $D_i$ by $r_i$; it might be that $r_i=0$. For a feasible solution $S$ we define $\alpha(S)=\sum r_i^2$. Since the total area of the disks in $S$ is $\pi\cdot\alpha(S)$, hereafter, we refer to $\alpha(S)$ as the total area of disks in $S$. We call $D_i$ a {\em full disk} if it has $p_{i-1}$ or $p_{i+1}$ on its boundary, a {\em zero disk} if its radius is zero, and a {\em partial disk} otherwise. For two points $p_i$ and $p_j$, we denote the Euclidean distance between $p_i$ and $p_j$ by $\dist{p_i}{p_j}$.

We briefly review the $O(n^2)$-time algorithm of Acharyya~\etal~\cite{Acharyya2017a}. First, compute a set $\mathcal{D}$ of disks centered at points of $P$, which is the superset of every optimal solution. For every disk $D\in \mathcal{D}$, that is centered at a point $p\in P$, define a weighted interval $I$ whose length is $2r$, where $r$ is the radius of $D$, and whose center is $p$. Set the weight of $I$ to be $r^2$. Let $\mathcal{I}$ be the set of these intervals. The disks corresponding to the intervals in a maximum weight independent set of the intervals in $\mathcal{I}$ forms an optimal solution to Problem 1. By construction, these disks are nonoverlapping, centered at $p_1,\dots, p_n$, and maximize the total area. Since the maximum weight independent set of $m$ intervals that are given in sorted order of their left endpoints can be computed in $O(m)$ time \cite{Hsiao1992}, the time complexity of the above algorithm is essentially dominated by the size of $\mathcal{D}$. Acharyya~\etal~\cite{Acharyya2017a} showed how to compute such a set $\mathcal{D}$ of size $\Theta(n^2)$ and order the corresponding intervals in $O(n^2)$ time. Therefore, the total running time of their algorithm is $O(n^2)$. 

We show how to improve the running time to $O(n)$. In fact we show how to find a set $\mathcal{D}$ of size $\Theta(n)$ and order the corresponding intervals in $O(n)$ time, provided that the points of $P$ are given in sorted order.  

\subsection{Computation of $\mathcal{D}$}
In this section we show how to compute a set $\mathcal{D}$ with a linear number of disks such that every disk in an optimal solution for Problem 1 belongs to $\mathcal{D}$.

Our set $\mathcal{D}$ is the union of three sets $F$, $\overrightarrow{D}$, and $\overleftarrow{D}$ of disks that are computed as follows. The set $F$ contains $2n$ disks representing the full disks and zero disks that are centered at points of $P$. We compute $\overrightarrow{D}$ by traversing the points of $P$ from left to right as follows; the computation of $\overleftarrow{D}$ is symmetric. For each point $p_i$ with $i\in\{2,\dots,n-1\}$ we define its {\em signature} $s(p_i)$ as 
$$
s(p_i) =
\begin{cases}
+       & \quad \text{if~~} \dist{p_{i-1}}{p_i}\leqslant \dist{p_i}{p_{i+1}}\\
-  & \quad \text{if~~} \dist{p_{i-1}}{p_i}> \dist{p_i}{p_{i+1}}.\\
\end{cases}
$$
Set $s(p_1)=-$ and $s(p_n)=+$.
We refer to the sequence $\mathcal{S}=s(p_1),\dots,s(p_{n})$ as the {\em signature sequence} of $P$. Let $\Delta$ be the multiset that contains all contiguous subsequences $s(p_i),\dots,s(p_j)$ of $\mathcal{S}$, with $i<j$, such that $s(p_i)=s(p_j)=-$, and $s(p_k)=+$ for all $i< k< j$; if $j=i+1$, then there is no $k$. For example, if $\mathcal{S}=-++-+++---+--++$, then $\Delta=\{-++-, -+++-,--, --, -+-, --\}$. Observe that for every sequence $s(p_i),\dots,s(p_j)$ in $\Delta$ we have that 
$$\dist{p_i}{p_{i+1}}\leqslant \dist{p_{i+1}}{p_{i+2}}\leqslant \dist{p_{i+2}}{p_{i+3}}\leqslant\dots\leqslant \dist{p_{j-1}}{p_j},\quad\text{and}\quad \dist{p_{j-1}}{p_j}>\dist{p_j}{p_{j+1}}.$$

Every plus sign in $\mathcal{S}$ belongs to at most one sequence in $\Delta$, and every minus sign in $\mathcal{S}$ belongs to at most two sequences in $\Delta$. Therefore, the size of $\Delta$ (the total length of its sequences) is at most $2n$. For each sequence $s(p_i),\dots,s(p_j)$ in $\Delta$ we add some disks to $\overrightarrow{D}$ as follows. Consider the full disk $D_j$ at $p_j$. Iterate on $k=j-1,j-2,\dots, i$. In each iteration, consider the disk $D_k$ that is centered at $p_k$ and touches $D_{k+1}$. If $D_{k}$ does not contain $p_{k-1}$ and its area is not larger than the area of $D_{k+1}$, then add $D_k$ to $\overrightarrow{D}$ and proceed to the next iteration, otherwise, terminate the iteration. See Figure~\ref{increasing-sequence-fig}. This finishes the computation of $\overrightarrow{D}$. Notice that $\overrightarrow{D}$ contains at most $n-1$ disks. The computation of $\overleftarrow{D}$ is symmetric; it is done in a similar way by traversing the points from right to left (all the $+$ signatures become $-$ and vice versa).

\begin{figure}[htb]
	\centering
	\includegraphics[width=.9\columnwidth]{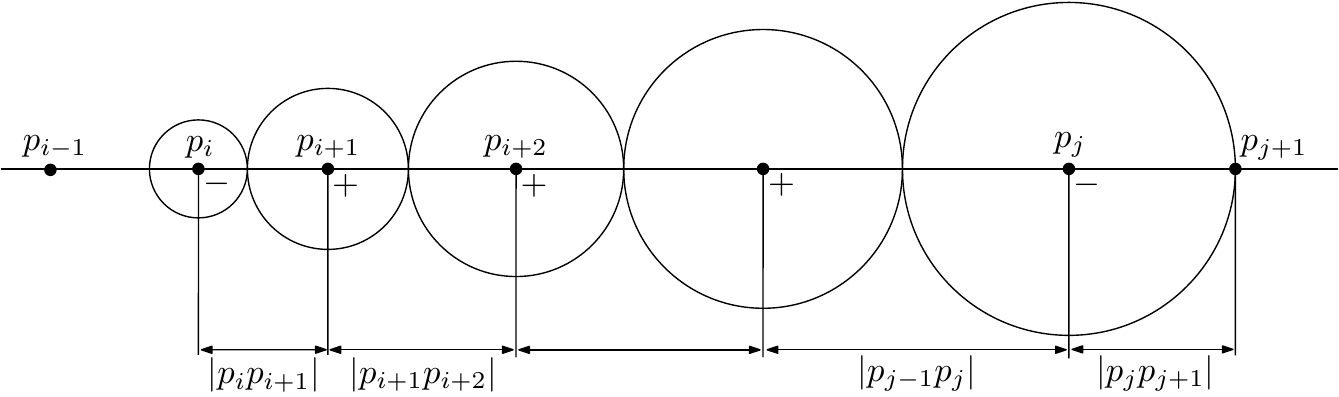}
	\caption{Illustration of a sequence $s(p_i),\dots, s(p_j)=-+++-$ in $\Delta$; construction of \text{$\protect\overrightarrow{D}$}.}
	\label{increasing-sequence-fig}
\end{figure}

The number of disks in $\mathcal{D}= F\cup\overrightarrow{D}\cup\overleftarrow{D}$ is at most $4n-2$. The signature sequence $\mathcal{S}$ can be computed in linear time. Having $\mathcal{S}$, we can compute the multiset $\Delta$, the disks in $\overrightarrow{D}$ as well as the corresponding intervals, as in \cite{Acharyya2017a} and described before, in sorted order of their left endpoints in total $O(n)$ time. Then the sorted intervals corresponding to circles in $\mathcal{D}$ can be computed in linear-time by merging the sorted intervals that correspond to sets $F$, $\overrightarrow{D}$, and $\overleftarrow{D}$. It remains to show that $\mathcal{D}$ contains an optimal solution for Problem 1. To that end, we first prove two lemmas about the structural properties of an optimal solution.

\begin{lemma}
	\label{first-last-points-lemma}
	Every feasible solution $S$ for Problem 1 can be converted to a feasible solution $S'$ where $D_1$ and $D_n$ are full disks and $\alpha(S')\geqslant \alpha(S)$.
\end{lemma}

\begin{proof}
	Recall that $n\geqslant 3$. We prove this lemma for $D_1$; the proof for $D_n$ is similar. Since $S$ is a feasible solution, we have that $r_1 + r_2 \leqslant \dist{p_1}{p_2}$. Let $S'$ be a solution that is obtained from $S$ by making $D_1$ a full disk and $D_2$ a zero disk. Since we do not increase the radius of $D_2$, it does not overlap $D_3$, and thus, $S'$ is a feasible solution. In $S'$, the radius of $D_1$ is $\dist{p_1}{p_2}$, and we have that $r_1^2+r_2^2\leqslant (r_1+r_2)^2\leqslant \dist{p_1}{p_2}^2$. This implies that $\alpha(S')\geqslant \alpha(S)$.
\end{proof}

\begin{lemma}
	\label{enlarge-lemma}
	If $D_i$, with $1<i<n$, is a partial disk in an optimal solution, then $r_i<\max(r_{i-1},r_{i+1})$.
\end{lemma}
\begin{proof}
	The proof is by contradiction; let $S$ be such an optimal solution for which $r_i\geqslant \max(r_{i-1},\allowbreak r_{i+1})$. First assume that $D_i$ touches at most one of $D_{i-1}$ and $D_{i+1}$. By slightly enlarging $D_i$ and shrinking its touching neighbor we can increase the total area of $S$. Without loss of generality suppose that $D_i$ touches $D_{i-1}$. Since $r_i\geqslant r_{i-1}$, $$(r_i+\epsilon)^2+(r_{i-1}-\epsilon)^2=r_i^2+r_{i-1}^2+2(r_i\epsilon-r_{i-1}\epsilon+\epsilon^2)>r_i^2+r_{i-1}^2>0,$$ for any $\epsilon>0$. This contradicts optimality of $S$. Now, assume that $D_i$ touches both $D_{i-1}$ and $D_{i+1}$, and that $r_{i-1}\leqslant r_{i+1}$. See Figure~\ref{enlarge-fig}. We obtain a solution $S'$ from $S$ by enlarging $D_i$ as much as possible, and simultaneously shrinking both $D_{i-1}$ and $D_{i+1}$. This makes $D_{i-1}$ a zero disk, $D_i$ a full disk, $D_{i+1}$ a zero or a partial disk, and does not change the other disks. The difference between the total areas of $S'$ and $S$ is 
	$$\left((r_i+r_{i-1})^2 + (r_{i+1}-r_{i-1})^2\right) -(r_{i-1}^2+r_i^2+r_{i+1}^2)=r_{i-1}^2+2r_{i-1}(r_i-r_{i+1})>0;$$
	this inequality is valid since $r_i\geqslant r_{i+1}\geqslant r_{i-1}>0$.
	This contradicts the optimality of $S$.
\end{proof}

\begin{figure}[htb]
	\centering
	\includegraphics[width=.8\columnwidth]{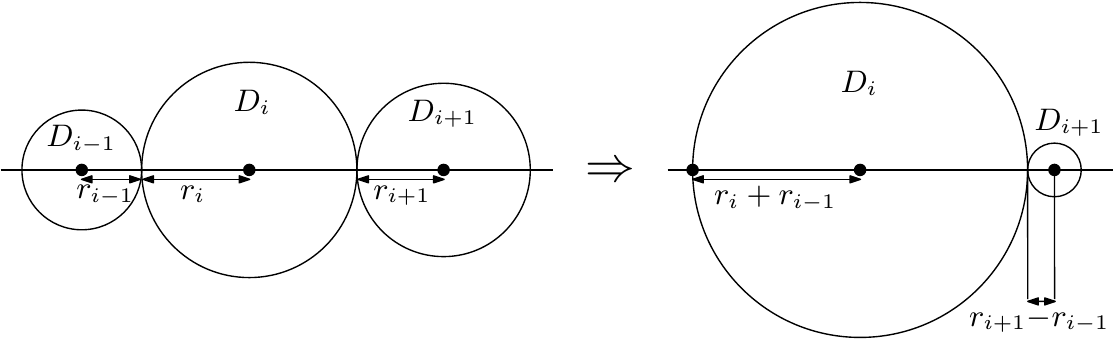}
	\caption{Illustration of the proof of Lemma~\ref{enlarge-lemma}.}
	\label{enlarge-fig}
\end{figure}

\begin{lemma}
	The set $\mathcal{D}$ contains an optimal solution for Problem 1.
\end{lemma}

\begin{proof}
	It suffices to show that every disk $D_k$, which is centered at $p_k$, in an optimal solution $S=\{D_1,\dots,D_n\}$ belongs to $\mathcal{D}$. By Lemma~\ref{first-last-points-lemma}, we may assume that both $D_1$ and $D_n$ are full disks. If $D_k$ is a full disk or a zero disk, then it belongs to $F$. Assume that $D_k$ is a partial disk. Since $S$ is optimal, $D_k$ touches at least one of $D_{k-1}$ and $D_{k+1}$, because otherwise we could enlarge $D_k$. 
	
	First assume that $D_k$ touches exactly one disk, say $D_{k+1}$. We are going to show that $D_k$ belongs to $\overrightarrow{D}$ (If $D_k$ touches only $D_{k-1}$, by a similar reasoning we can show that $D_k$ belongs to $\overleftarrow{D}$). Notice that $r_k<r_{k+1}$, because otherwise we could enlarge $D_k$ and shrink $D_{k+1}$ simultaneously to increase $\alpha(S)$, which contradicts the optimality of $S$. Since $D_k$ is partial and touches $D_{k+1}$, we have that $D_{k+1}$ is either full or partial. If $D_{k+1}$ is full, then it has $p_{k+2}$ on its boundary, and thus $s(p_{k+1})=-$. By our definition of $\Delta$, for some $i<k+1$, the sequence $s(p_i),\dots, s(p_{k+1})$ belongs to $\Delta$. Then by our construction of $\overrightarrow{D}$ both $D_{k+1}$ and $D_k$ belong to $\overrightarrow{D}$, where $k+1$ plays the role of $j$. Assume that $D_{k+1}$ is partial. Then $D_{k+2}$ touches $D_{k+1}$, because otherwise we could enlarge $D_{k+1}$ and shrink $D_{k}$ simultaneously to increase $\alpha(S)$. Recall that $r_k<r_{k+1}$. Lemma~\ref{enlarge-lemma} implies that $r_{k+1}<r_{k+2}$. This implies that $\dist{p_k}{p_{k+1}}<\dist{p_{k+1}}{p_{k+}}$, and thus $s(p_{k+1})=+$. Since $D_{k+1}$ is partial and touches $D_{k+2}$, we have that $D_{k+2}$ is either full or partial. If $D_{k+2}$ is full, then it has $p_{k+3}$ on its boundary, and thus $s(p_{k+2})=-$. By a similar reasoning as for $D_{k+1}$ based on the definition of $\Delta$ and $\overrightarrow{D}$, we get that $D_{k+2}$, $D_{k+1}$, and $D_k$ are in $\overrightarrow{D}$. If $D_{k+2}$ is partial, then it touches $D_{k+3}$ and again by Lemma~\ref{enlarge-lemma} we have $r_{k+2}<r_{k+3}$ and consequently $s(p_{k+2})=+$. By repeating this process, we stop at some point $p_j$, with $j\leqslant n-2$, for which $D_j$ is a full disk, $r_{j-1}<r_{j}$, and $s(p_{j})=-$; notice that such a $j$ exists because $D_n$ is a full disk and consequently $D_{n-1}$ is a zero disk. To this end we have that $s(p_k)\in\{+,-\}$, $s(p_j)=-$, and $s(p_{k+1}),\dots, s(p_{j-1})$ is a plus sequence. Thus, $s(p_k),\dots,s(p_j)$ is a subsequence of some sequence $s(p_i),\dots,s(p_j)$ in $\Delta$. Our construction of $\overrightarrow{D}$ implies that all disks $D_k, \dots, D_{j}$ belong to $\overrightarrow{D}$.

	Now assume that $D_k$ touches both $D_{k-1}$ and $D_{k+1}$. By Lemma~\ref{enlarge-lemma} we have that $D_k$ is strictly smaller than the largest of these disks, say $D_{k+1}$. By a similar reasoning as in the previous case we get that $D_k\in \overrightarrow{D}$. 
\end{proof}

\section{Problem 2: Client-Server Coverage with Minimum Radii}
\label{Problem-2-section}

In this section we study Problem 2: Let $P=\{p_1,\dots,p_n\}$ be a set of $n$ points on a straight-line $\ell$ that is partitioned into two sets, namely clients and servers. We want to assign to every server in $P$ a radius such that the disks with these radii cover all clients and the sum of their radii is as small as possible. Bil\`{o} \etal~\cite{Bilo2005} showed that this problem can be solved in polynomial time. Lev-Tov and Peleg \cite{Lev-Tov2005} showed how to obtain such an assignment in $O(n^3)$ time. Alt \etal~\cite{Alt2006} presented an O(n log n)-time 2-approximation algorithm for this problem. We show how to solve this problem optimally in $O(n^2)$ time. 

\begin{theorem}
	Given a total of $n$ collinear clients and servers, in $O(n^2)$ time, we can find a set of disks centered at servers that cover all clients and where the sum of the radii of the disks is minimum.
\end{theorem}

Without loss of generality assume that $\ell$ is horizontal, and that $p_1,\dots,p_n$ is the sequence of points of $P$ in increasing order of their $x$-coordinates. We refer to a disk with radius zero as a {\em zero disk}, to a set of disks centered at servers and covering all clients as a {\em feasible solution}, and to the sum of the radii of the disks in a feasible solution as its {\em cost}. We denote the radius of a disk $D$ by $r(D)$, and denote by $D(p,q)$ a disk that is centered at the point $p$ with the point $q$ on its boundary.

We describe a top-down dynamic programming algorithm that maintains a table $T$ with $n$ entries $T(1),\dots, T(n)$. Each table entry $T(k)$ represents the cost of an optimal solution for the subproblem that consists of points $p_1,\dots,p_k$. The optimal cost of the original problem will be stored in $T(n)$; the optimal solution itself can be recovered from $T$. In the rest of this section we show how to solve a subproblem $p_1,\dots,p_k$. In fact, we show how to compute $T(k)$ recursively by a top-down dynamic programming algorithm. 
To that end, we first describe our three base cases:
\begin{itemize}
	\item There is no client. In this case $T(k)=0$.
	\item There are some clients but no server. In this case $T(k)=+\infty$.
	\item There are some clients and exactly one server, say $s$. In this case $T(k)$ is the radius of the smallest disk that is centered at $s$ and covers all the clients.
\end{itemize}	

Assume that the subproblem $p_1,\dots, p_k$ has at least one client and at least two servers. We are going to derive a recursion for $T(k)$.

\begin{observation}
	\label{client-on-boundary-obs}
	Every disk in any optimal solution has a client on its boundary. 
\end{observation}

\begin{lemma}
	\label{server-in-disk-lemma}
	No disk contains the center of some other non-zero disk in an optimal solution. 
\end{lemma}

\begin{proof}
	Our proof is by contradiction. Let $D_i$ and $D_j$ be two disks in an optimal solution such that $D_i$ contains the center of $D_j$. Let $p_i$ and $p_j$ be the centers of $D_i$ and $D_j$, respectively, and $r_i$ and $r_j$ be the radii of $D_i$ and $D_j$, respectively. See Figure~\ref{server-in-disk-fig}(a). Since $D_i$ contains $p_j$, we have $r_i>\dist{p_i}{p_j}$. Let $D'_i$ be the disk of radius $\dist{p_i}{p_j}+r_j$ that is centered at $p_i$. Notice that $D'_i$ covers all the clients that are covered by $D_i\cup D_j$. By replacing $D_i$ and $D_j$ with $D'_i$ we obtain a feasible solution whose cost is smaller than the optimal cost, because $\dist{p_i}{p_j}+r_j< r_i+r_j$. This contradicts the optimality of the initial solution.  
\end{proof}

\begin{figure}[htb]
	\centering
	\setlength{\tabcolsep}{0in}
	$\begin{tabular}{cc}
	\multicolumn{1}{m{.38\columnwidth}}{\centering\includegraphics[width=.33\columnwidth]{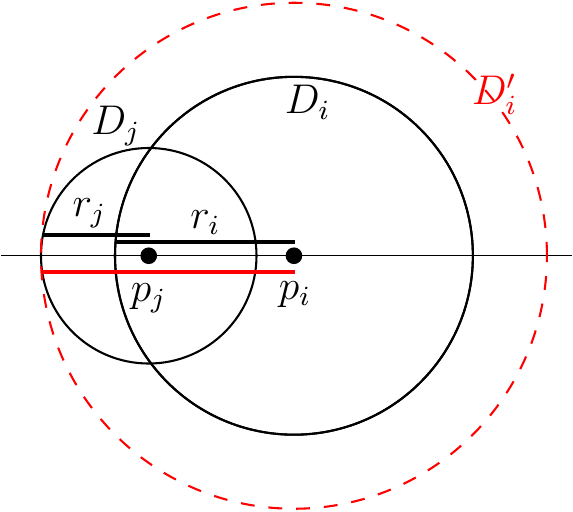}}
	&\multicolumn{1}{m{.62\columnwidth}}{\centering\includegraphics[width=.57\columnwidth]{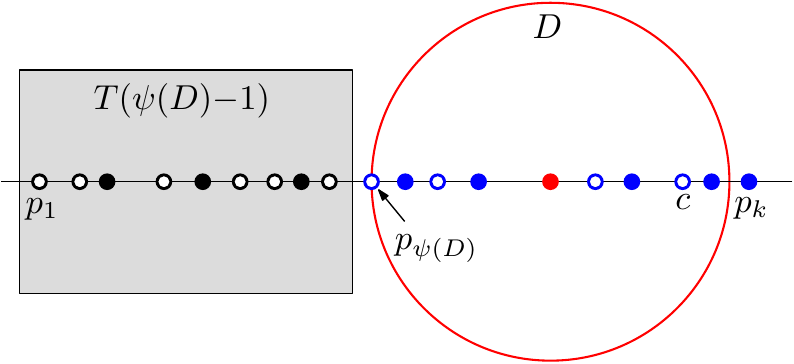}}
	\\
	(a) & (b)
	\end{tabular}$
	\caption{(a) Illustration of the proof of Lemma~\ref{server-in-disk-lemma}. (b) Clients are shown by small circles, and servers are shown by small disks. $p_{\psi(D)}$ is the leftmost point (client or server) in $D$.}
	\label{server-in-disk-fig}
\end{figure}

Let $c$ be the rightmost client in $p_1,\dots, p_k$. For a disk $D$ that covers $c$, let $\psi(D)\in \{1,\dots,k\}$ be the smallest index for which the point $p_{\psi(D)}$ is in the interior or on the boundary of $D$, i.e., $\psi(D)$ is the index of the leftmost point of $p_1,\dots,p_k$ that is in $D$. See Figure~\ref{server-in-disk-fig}(b). 

We claim that only one disk in an optimal solution can cover $c$, because, if two disks cover $c$ then if their centers lie on the same side of $c$, we get a contradiction to Lemma~\ref{server-in-disk-lemma}, and if their centers lie on different sides of $c$, then by removing the disk whose center is to the right of $c$ we obtain a feasible solution with smaller cost. 
Let $S^*$ be an optimal solution (with minimum sum of the radii) that has a maximum number of non-zero disks. 
Let $D^*$ be the disk in $S^*$ that covers $c$. All other clients in $p_{\psi(D^*)},\dots,p_k$ are also covered by $D^*$, and thus, they do not need to be covered by any other disk. As a consequence of Lemma~\ref{server-in-disk-lemma}, the servers that are in $D^*$ and the servers that lie to the right of $D^*$ cannot be used to cover any clients in $p_1,\dots, p_{\psi(D^*)-1}$. Therefore, if we have $D^*$, then the problem reduces to a smaller instance that consists of the points to the left of $D^*$, i.e., $p_1,\dots, p_{\psi(D^*)-1}$. See Figure~\ref{server-in-disk-fig}(b). Thus, the cost of the optimal solution for the subproblem $p_1,\dots, p_k$ can be computed as
$T(k)=T(\psi(D^*)-1) + r(D^*)$. 

In the rest of this section we compute a set $\mathcal{D}_k$ of $O(k)$ disks each of them covering $c$. Then we claim that $D^*$ belongs to $\mathcal{D}_k$. Having $\mathcal{D}_k$, we can compute $T(k)$ by the following recursion:
$$T(k)=\min\{T(\psi(D)-1) + r(D): D\in \mathcal{D}_k\}.$$ 
\changed{
Now we show how to compute $\mathcal{D}_k$. Recall from Observation~\ref{client-on-boundary-obs} that every disk in the optimal solution (including $D^*$) contains a client on its boundary. Using this observation, we compute $\mathcal{D}_k$ in two phases. In the first phase, for every server $s$ we add the disk $D(s,c)$ to $\mathcal{D}_k$.
In the second phase, for every client $c'$, with $c' \neq c$, we add a disk $D(s',c')$ to $\mathcal{D}_k$, where $s'$ is the first server to the right side of the midpoint of segment $cc'$. Since for every server and for every client (except for $c$) we add one disk to $\mathcal{D}_k$, this set has at most $k-1$ disks. The disks that we add in phase one can be computed in $O(k)$ time by sweeping the servers from right to left. The disks that we add in phase two can also be computed in $O(k)$ time by sweeping the clients from right to left, using this property that the server $s'$ associated with the next client $c'$ is on or to the left side of the server associated with the current client. Hence, the set $\mathcal{D}_k$, and consequently the entry $T(k)$, can be computed in $O(k)$ time. Therefore, our dynamic programming algorithm computes all entries of $T$ in $O(n^2)$ time.

One final issue we need to address is the correctness of our algorithm, which is to show that $D^*$ belongs to $\mathcal{D}_k$. Let $s^*$ be the server that is the center of $D^*$ and let $c^*$ be the client on the boundary of $D^*$ (such a client exists by Observation~\ref{client-on-boundary-obs}). Recall that $D^*$ covers the rightmost client $c$. If $c^*=c$, then $D^*$ has been added to $\mathcal{D}_k$ in the phase one.  Assume that $c^*\neq c$. In this case $c^*$ is the left intersection point of the
boundary of $D^*$ with $\ell$ because $c^*$ is to the left side of $c$. Let $m$ be the mid point of the line segment $c^*c$, and let $s$ be the first server to the right of $m$. The server $s^*$ cannot be to the left side of $m$ because otherwise $D^*$ could not cover $c$. Also, $s^*$ cannot be to the right side of $s$ because otherwise the disk $D(s, c^*)$, which is smaller than $D^*$, covers the same set of clients as $D^*$ does, in particular it covers $c^*$ and $c$. Therefore, we have $s^*=s$, and thus $D^*=D(s,c^*)$, which has been added to $\mathcal{D}_k$ in phase two. This finishes the proof of correctness of our algorithm. 
}

\section{Problem 3: Point-Interval Coverage with Minimum Area}
\label{Problem-3-section}
Let $I=[a,b]$ be an interval on the $x$-axis in the plane. We say that a set of disks {\em covers} $I$ if $I$ is a subset of the union of the disks in this set. \changed{Let $P=\{p_1,\dots,p_n\}$ be a set of $n$ points on $I$, that are ordered from left to right, and such that $p_1=a$ and $p_n=b$. A {\em point-interval coverage} for the pair $(P,I)$ is a set $S=\{D_1,\dots,D_n\}$ of $n$ disks that cover $I$ such that for every $i\in\{1,\dots,n\}$ the disk $D_i$ contains the point $p_i$, i.e., $p_i$ is in the interior or on the boundary of $D_i$.} See Figure~\ref{setting-fig}. The {\em point-interval coverage} problem is to find such a set of disks with minimum total area. In this section we show how to solve this problem in $O(n^2)$ time.

\begin{figure}[htb]
	\centering
	\includegraphics[width=.6\columnwidth]{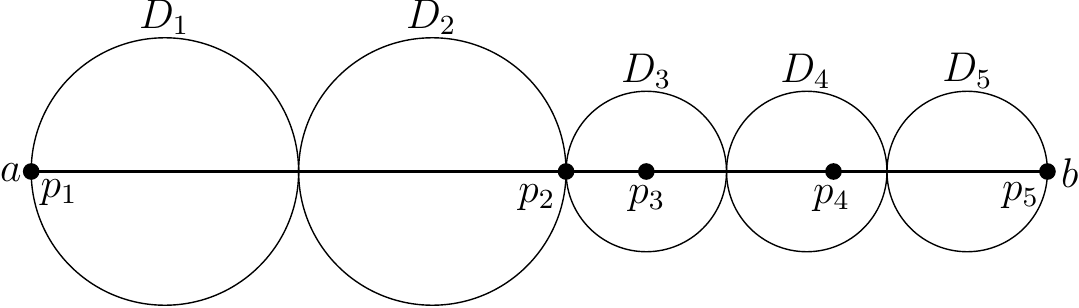}
	\caption{The minimum point-interval coverage for $(\{p_1,\dots,p_5\},[a,b])$; \changed{for every $i$, $D_i$ contains $p_i$}.}
	\label{setting-fig}
\end{figure}

\begin{theorem}
	Given $n$ points on an interval, in $O(n^2)$ time, we can find a set of disks covering
	the entire interval such that every disk contains at least one point and where the total area of the disks is minimum.
\end{theorem}

If we drop the condition \changed{that $D_i$ should contain $p_i$}, then the problem can be solved in linear time by using Observation~\ref{lower-bound-obs} (which is stated below). 
First we prove some lemmas about the structural properties of an optimal point-interval coverage. We say that a disk is {\em anchored} at a point $p$ if it has $p$ on its boundary.
We say that two intersecting disks {\em touch} each other if their intersection is exactly one point, and we say that they {\em overlap} otherwise.

\begin{figure}[htb]
	\centering
	\setlength{\tabcolsep}{0in}
	$\begin{tabular}{cc}
	\multicolumn{1}{m{.5\columnwidth}}{\centering\includegraphics[width=.35\columnwidth]{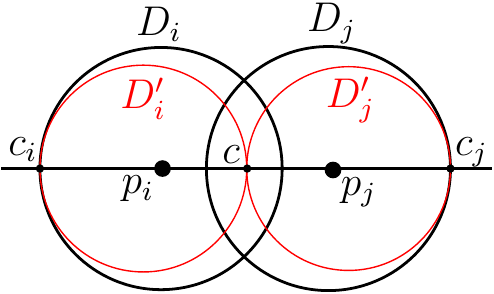}}
	&\multicolumn{1}{m{.5\columnwidth}}{\centering\includegraphics[width=.35\columnwidth]{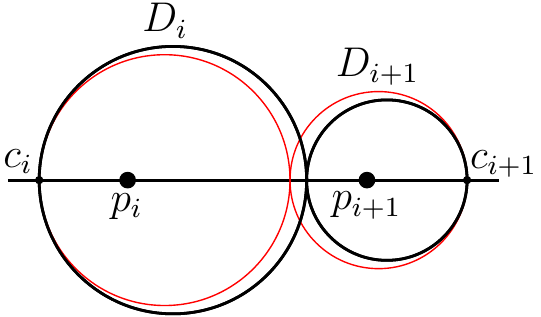}}
	\\
	(a) & (b) 
	\end{tabular}$
	\caption{Illustrations of the proofs of (a) Lemma~\ref{non-overlap-lemma}, and (b) Lemma~\ref{equality-lemma}.}
	\label{touching-fig}
\end{figure}

\begin{lemma}
	\label{non-overlap-lemma}
	There is no pair of overlapping disks in any optimal solution for the point-interval coverage problem. 
\end{lemma}

\begin{proof}
	\changed{Our proof is by contradiction. Consider two overlapping disks $D_i$ and $D_j$, with $i<j$, in an optimal solution. Since $D_i$ contains $p_i$ and $D_j$ contains $p_j$, there exists a point $c$ on the line segment $p_ip_j$ that is in $D_i\cap D_j$. Let $c_i$ and $c_j$ be the leftmost and the rightmost points of the interval that is covered by $D_i\cup D_j$; see Figure~\ref{touching-fig}(a). Let $D'_i$ and $D'_j$ be the disks with diameters $c_ic$ and $cc_j$, respectively. The areas of $D'_i$ and $D'_j$ are smaller than the areas of $D_i$ and $D_j$, respectively. Moreover, $D'_i$ contains $p_i$, $D'_j$ contains $p_j$, and $D'_i\cup D'_j$ covers the same interval $[c_i,c_j]$ as $D_i\cup D_j$ does. Therefore, by replacing $D_i$ and $D_j$ with $D'_i$ and $D'_j$ we obtain a solution whose total area is smaller than the optimal area, which is a contradiction.}
\end{proof}

\begin{lemma}
	\label{equality-lemma}
	In any optimal solution, if the intersection point of $D_i$ and $D_{i+1}$ does not belong to $P$, then $D_i$ and $D_{i+1}$ have equal radius. 
\end{lemma}
\begin{proof}
	Let $c$ be the intersection point of $D_i$ and $D_{i+1}$. Let $c_i$ be the left intersection point of the boundary of $D_i$ with the $x$-axis, and let $c_{i+1}$ be the right intersection point of the boundary of $D_{i+1}$ with the $x$-axis; see Figure~\ref{touching-fig}(b). We proceed by contradiction, and assume, without loss of generality, that $D_{i+1}$ is smaller than $D_i$. We shrink $D_i$ (while anchored at $c_i$) and enlarge $D_{i+1}$ (while anchored at $c_{i+1}$) simultaneously by a small value. This gives a valid solution whose total area is smaller than the optimal area, because our gain in the area of $D_{i+1}$ is smaller than our loss from the area of $D_i$. This contradicts the optimality of our initial solution.
\end{proof}

The following lemma and observation play important roles in our algorithm for the point-interval coverage problem, which we describe later.

\begin{lemma}
	\label{equity-lemma}
	Let $R>0$ be a real number, and $r_1,r_2,\ldots,r_k$ be a sequence of positive real numbers such that 
	$\sum_{i=1}^k r_i = R$. Then 
	\begin{equation}  \label{eq1} 
	\sum_{i=1}^k r_i^2 \geq \sum_{i=1}^k (R/k)^2 = R^2/k , 
	\end{equation} 
	i.e., the sum on the left-hand side of \eqref{eq1} is minimum if all 
	$r_i$ are equal to $R/k$. 
\end{lemma} 
\begin{proof}
	If $f$ is a convex function, then\textemdash by Jensen’s inequality\textemdash we have 
	$$f \left( \sum_{i=1}^k \frac{r_i}{k} \right) \leqslant
	\sum_{i=1}^k \frac{f(r_i)}{k} . 
	$$  
	Since the function $f(x) = x^2$ is convex, 
	it follows that 
	$$\left( \frac{R}{k} \right)^2 = f \left( \frac{R}{k} \right)
	= f \left( \sum_{i=1}^k \frac{r_i}{k} \right)\leqslant \sum_{i=1}^k \frac{r_i^2}{k}, 
	$$
	which, in turn, implies Inequality~\eqref{eq1}.
\end{proof} 

The minimum sum of the radii of a set of disks that cover $I=[a,b]$ is $|ab|/2$. The following observation is implied by Lemma~\ref{equity-lemma}, by setting $R=|ab|/2$ and $k=n$. 

\begin{observation}
	\label{lower-bound-obs}
	The minimum total area of $n$ disks covering $I$ is obtained by a sequence of $n$ disks of equal radius such that every two consecutive disks touch each other; see Figure~\ref{unity-fig}.
\end{observation}  

\begin{figure}[htb]
	\centering
	\includegraphics[width=.7\columnwidth]{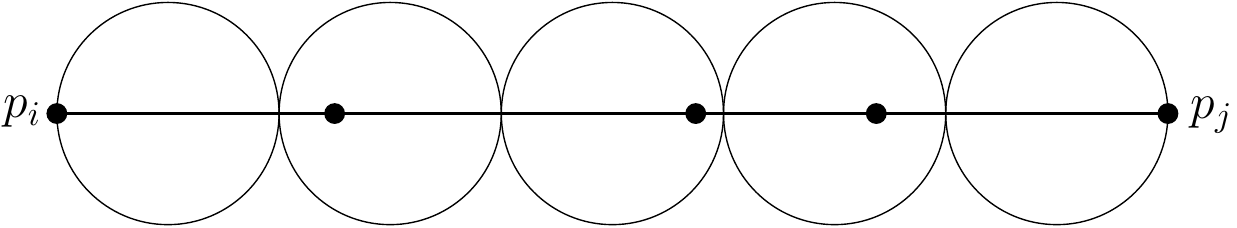}
	\caption{A valid unit-disk covering.}
	\label{unity-fig}
\end{figure}

We refer to the covering of $I$ that is introduced in Observation~\ref{lower-bound-obs} as the {\em unit-disk covering} of $I$ with $n$ disks. Such a covering is called {\em valid} if it is a point-interval coverage for $(P,I)$. 

\subsection{A Dynamic-Programming Algorithm}
\label{dp-subsection}
In this subsection we present an $O(n^3)$-time dynamic-programming algorithm for the point-interval coverage problem. In Subsection~\ref{improved-subsection} we show how to improve the running time to $O(n^2)$.

First, we review some properties of an optimal solution for the point-interval coverage problem that enable us to present a top-down dynamic programming algorithm.
Let $C^*=D_1,\dots,D_n$ be the sequence of $n$ disks in an optimal solution for this problem. Recall that as a consequence of Lemma~\ref{non-overlap-lemma}, the intersection of every two consecutive disks in $C^*$ is a point. If there is no $k\in\{1,\dots, n-1\}$ for which the intersection point of $D_k$ and $D_{k+1}$ belongs to $P$, then Lemma~\ref{equality-lemma} implies that all disks in $C^*$ have equal radius, and thus, $C^*$ is a valid unit-disk covering. Assume that for some $k\in\{1,\dots, n-1\}$ the intersection point of $D_k$ and $D_{k+1}$ is a point $p\in P$. Notice that either $p=p_k$ or $p=p_{k+1}$. In either case, $C^*$ is the union of the optimal solutions for two smaller problem-instances $(P_1,I_1)$ and $(P_2,I_2)$ where $I_1=[a,p]$, $I_2=[p,b]$, $P_1=\{p_1,\dots,p_k\}$ and $P_2=\{p_{k+1},\dots,p_n\}$.

\begin{figure}[htb]
	\centering
	\includegraphics[width=.73\columnwidth]{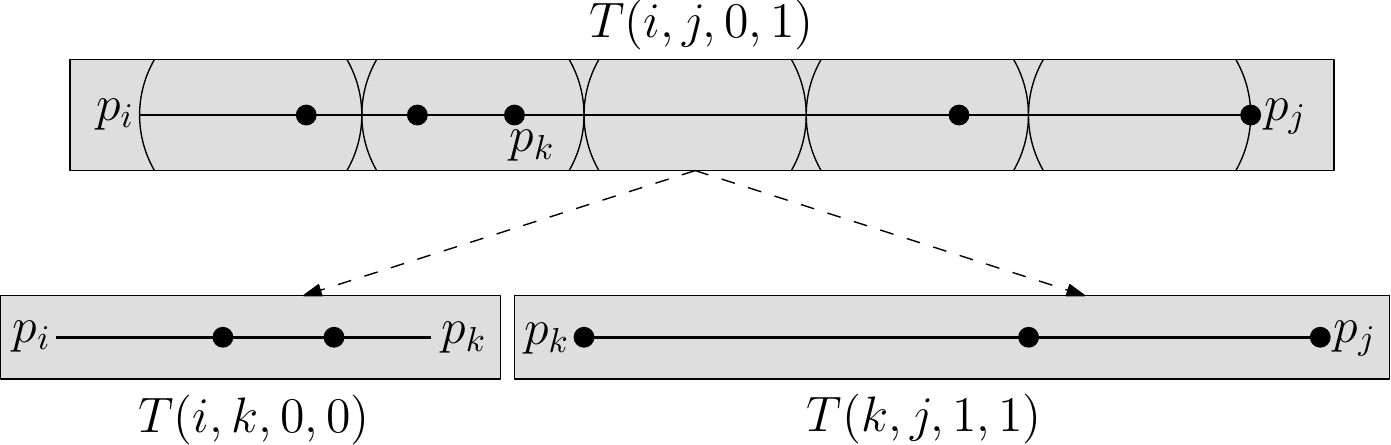}
	\caption{An instance for which the unit-disk covering (shown on the top interval) is not valid.}
	\label{split-fig}
\end{figure}

We define a subproblem $(P_{ij},I_{ij})$ and represent it by four indices $(i,j,i',j')$ where $1\leqslant i<j\leqslant n$ and $i',j'\in\{0,1\}$. The indices $i$ and $j$ indicate that $I_{ij} =[p_i,p_j]$. The set $P_{ij}$ contains the points of $P$ that are on $I_{ij}$ provided that $p_i$ belongs to $P_{ij}$ if and only if $i'=1$ and $p_j$ belongs to $P_{ij}$ if and only if $j'=1$. For example, if $i'=1$ and $j'=0$, then $P_{ij}=\{p_i,p_{i+1},\dots, p_{j-1}\}$. We define $T(i,j,i',j')$ to be the cost (total area) of an optimal solution for subproblem $(i,j,i',j')$. The optimal cost of the original problem will be stored in $T(1,n,1,1)$. We compute $T(i,j,i',j')$ as follows. If the unit-disk covering is a valid solution for $(i,j,i',j')$, then by Observation~\ref{lower-bound-obs} it is optimal, and thus we assign its total area to $T(i,j,i',j')$. Otherwise, as we discussed earlier, there is a point $p_k$ of $P$ with $k\in\{i+1,\dots, j-1\}$ that is the intersection point of two consecutive disks in the optimal solution. This splits the problem into two smaller subproblems, one to the left of $p_k$ and one to the right of $p_k$. The point $p_k$ is assigned either to the left subproblem or to the right subproblem. See Figure~\ref{split-fig} for an instance in which the unit-disk covering is not valid, and $p_k$ is assigned to the right subproblem. In the optimal solution, $p_k$ is assigned to the one that minimizes the total area, which is 
$$T(i,j,i',j')=\min\{T(i,k,i',1)+T(k,j,0,j'), T(i,k,i',0)+T(k,j,1,j')\}.$$
Since we do not know the value of $k$, we try all possible values and pick the one that minimizes $T(i,j,i',j')$.

There are three base cases for the above recursion. (1) No point of $P$ is assigned to the current subproblem: we assign $+\infty$ to $T(\cdot)$, which implies this solution is not valid. (2) Exactly one point of $P$ is assigned to the current subproblem: we cover $[p_i,p_j]$ with one disk of diameter $|p_ip_j|$ and assign its area to $T(\cdot)$. (3) More than one point of $P$ is assigned to the current subproblem and the unit-disk covering is valid: we assign the total area of this unit-disk covering to $T(\cdot)$.

The total number of subproblems is at most $2\cdot 2\cdot {n \choose 2} = O(n^2)$, because $i$ and $j$ take ${n \choose 2}$ different values, and each of $i'$ and $j'$ takes two different values. 
The time to solve each subproblem $(i,j,i',j')$ is proportional to the time for checking the validity of the unit-disk covering for this subproblem plus the iteration of $k$ from $i+1$ to $j-1$; these can be done in total time $O(j-i)$. Thus, the running time of our dynamic programming algorithm is $O(n^3)$. 

In the next section we present a more involved dynamic-programming algorithm that improves the running time to $O(n^2)$. Essentially, our algorithm verifies the validity of the unit-disk coverings for all subproblems $p_i,\dots,p_j$ in $O(n^2)$ time.

\subsection{Improving the running time}
\label{improved-subsection}
We describe a top-down dynamic programming algorithm that maintains a table $T$ with $2n$ entries $T(j,j')$ where $j\in\{1,\dots, n\}$ and $j'\in\{0,1\}$. Each entry $T(j, j')$ represents the cost of an optimal solution for the subproblem that consists of interval $I_j=[p_1,p_j]$ and a point set $P_j$. 
If $j' = 1$, then $P_j = I_j \cap P$, whereas, if $j' = 0$, then $P_j = I_j \cap P \setminus \{ p_j \}$. 
The optimal cost of the original problem will be stored in $T(n, 1)$; the optimal solution itself can be recovered from $T$. In the rest of this section we show how to solve subproblem $(j,j')$. If the unit-disk covering is a valid solution for $(j,j')$, then we assign its total area to $T(j,j')$. Otherwise, there must be some point $p_k\in P$ with $k\in\{2,\dots, j-1\}$ that is the intersection point of two consecutive disks in the optimal solution. Let $i$ be the largest such $k$. This choice of $i$ implies that the interval $[p_i,p_j]$ is covered by unit disks, and thus, we only need to solve the subproblem to the left of $p_i$ optimally for two cases where $i'=0$ and $i'=1$. Let $U(i,j,i'j')$ denote the cost of a unit-disk covering for the problem instance $(i,j,i',j')$ (that is defined in the previous section). Then

$$T(j,j')=\min\{T(i,1)+U(i,j,0,j'), T(i,0)+U(i,j,1,j')\}.$$

Since we do not know the value of $i$, we try all possible values and pick the one
that minimizes $T(j,j')$. 

The total number of subproblems is $2n$, and the time to solve each subproblem $(j,j')$ is proportional to the total time for the iterations of $i$ from $2$ to $j-1$ plus the time for computing the unit-disk covering for $(i,j,i',j')$ and checking its validity. Let $u(j)$ denote the time for computing and checking validity of unit-disk coverings for all $i$. Then the time to compute $T(j,j')$ is $O(j)+u(j)$. Therefore, the running time of our algorithm, i.e. the time to compute $T(n,1)$, is $$\sum_{j=1}^n{O(j)+u(j)}=O(n^2)+\sum_{j=1}^n{u(j)}=O(n^2)+\sum_{j=1}^n\sum_{i=2}^{j-1}{u(i,j,i',j')},$$
where $u(i,j,i',j')$ denotes the time of computing the unit-disk covering for $(i,j,i',j')$ and checking its validity. In the rest of this section we will show how to do this, for all $(i,j,i',j')$, in $O(n^2)$ time. This implies that the total running time of our algorithm is $O(n^2)$.

\begin{figure}[htb]
	\centering
	\includegraphics[width=.7\columnwidth]{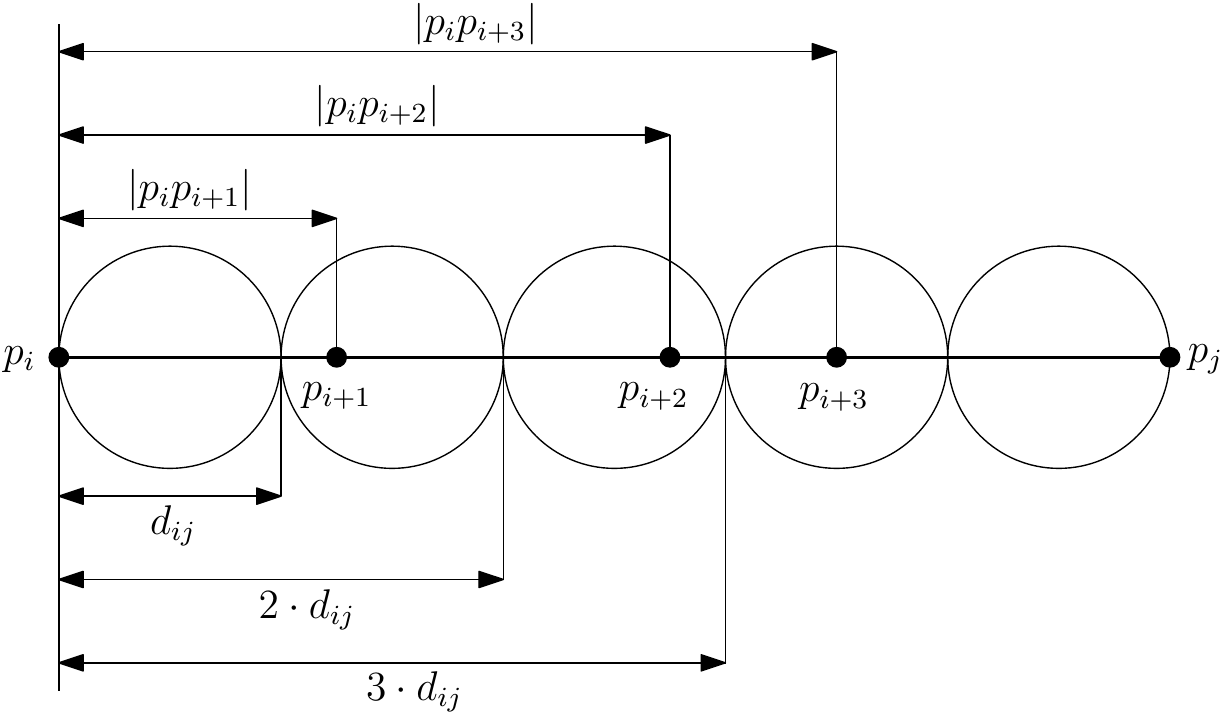}
	\caption{The validity of the unit-disk covering for $(i,j,1,1)$.}
	\label{unit-disk-validity-fig}
\end{figure}

Take any $i\in\{1,\dots, n-1\}$. We show how to check the validity of the unit-disk covering for $(i,j,i',j')$, where $j$ iterates from $i+1$ to $n$. We are going to show how to do this in $O(n-i)$ total time, for all values of $j$. This  will imply that we can check the validity of the unit-disk coverings for all $i$ and $j$ in $O(n^2)$ time. We describe the procedure for the case when $i'=1$ and $j'=1$; the other three cases can be handled similarly. Recall that $I_{ij}=[p_i,p_j]$ is the interval and $P_{ij}=\{p_i,\dots,p_j\}$ is the point set that are associated with $(i,j,i',j')$. Notice that the length of $I_{ij}$ is $|p_ip_j|$, and the number of points in $P_{ij}$ is $n_{ij}=j-i+1$.
Then the diameter of each disk in the unit-disk covering of $(i,j,1,1)$ is $d_{ij}=|p_ip_j|/n_{ij}$.
In order to have a valid unit-disk covering for $(i,j,1,1)$, the following conditions are necessary and sufficient (see Figure~\ref{unit-disk-validity-fig})

\begin{align}
\notag & |p_ip_{i+1}|~~\leqslant d_{ij}\\
\notag d_{ij}\leqslant~~  & |p_ip_{i+2}|~~\leqslant 2\cdot d_{ij}\\
\notag 2\cdot d_{ij}\leqslant~~  & |p_ip_{i+3}|~~\leqslant 3\cdot d_{ij}\\
\notag 3\cdot d_{ij}\leqslant~~  & |p_ip_{i+4}|~~\leqslant 4\cdot d_{ij}\\
\notag &~~~~ \vdots\\
\notag  (n_{ij}-1)\cdot d_{ij}\leqslant~~  &~ |p_ip_{j}|~~~\leqslant n_{ij}\cdot d_{ij}.
\end{align}
The above inequalities are equivalent to the following two inequalities
\begin{align}
\label{eq3} d_{ij}\geqslant&\max\left\{|p_ip_{i+1}|,\frac{|p_ip_{i+2}|}{2},\frac{|p_ip_{i+3}|}{3},\dots, \frac{|p_ip_j|}{n_{ij}}\right\},\\
\label{eq4}d_{ij}\leqslant&\min\left\{|p_ip_{i+2}|,\frac{|p_ip_{i+3}|}{2},\frac{|p_ip_{i+4}|}{3},\dots, \frac{|p_ip_j|}{n_{ij}-1}\right\}.
\end{align} 

Let $M(i,j)$ denote the maximum value in Inequality~\eqref{eq3}, and let $m(i,j)$ denote the minimum value in Inequality~\eqref{eq4}. To check the validity of the unit-disk covering for $(i,j,1,1)$, it suffices to compare $d_{ij}$ with these two values. Recall that $i$ is fixed and $j$ iterates from $i+1$ to $n$. Now we show how to check the validity of the unit-disk covering for subproblem $(i,j+1,1,1)$, in  $O(1)$  time. In this subproblem, the diameter of the unit disks is $d_{i(j+1)}$ and the number of points is $n_{i(j+1)}=n_{ij}+1$; these values can be computed in  $O(1)$  time. In order to have a valid unit-disk covering for $(i,j+1,1,1)$, the following two inequalities are necessary and sufficient
\begin{align}
\notag d_{i(j+1)}\geqslant&\max\left\{|p_ip_{i+1}|,\frac{|p_ip_{i+2}|}{2},\frac{|p_ip_{i+3}|}{3},\dots, \frac{|p_ip_j|}{n_{ij}},\frac{|p_ip_{j+1}|}{n_{ij}+1}\right\},\\
\notag d_{i(j+1)}\leqslant&\min\left\{|p_ip_{i+2}|,\frac{|p_ip_{i+3}|}{2},\frac{|p_ip_{i+4}|}{3},\dots, \frac{|p_ip_j|}{n_{ij}-1}, \frac{|p_ip_{j+1}|}{n_{ij}}\right\}.
\end{align}

Thus, we can compute 
$$
\notag M(i,j+1)=\max\left\{M(i,j), \frac{|p_ip_{j+1}|}{n_{ij}+1}\right\}\text{,~~~~and~~~~}
\notag m(i,j+1)=\min\left\{m(i,j), \frac{|p_ip_{j+1}|}{n_{ij}}\right\},
$$

in  $O(1)$  time. Then the unit-disk covering is valid for $(i,j+1,1,1)$ if and only if 
$m(i,j+1)\leqslant d_{i(j+1)}\leqslant M(i,j+1)$; this can be verified in  $O(1)$  time. Thus, by keeping $M(i,j)$ and $m(i,j)$ from the previous iteration, we can check the validity of the unit-disk covering for the current iteration, in  $O(1)$  time. Therefore, we can check the validity of $(i,j,1,1)$ for all $j\in\{i+1,\dots,n\}$ in $O(n-i)$ total time. This finishes the proof.

\section{Conclusion: An Open Problem}
We considered three optimization problems on collinear points in the plane. Here we present a related open problem: given a set of collinear points, we want to assign to each point a disk, centered at that point, such that the underlying disk graph is connected and the sum of the areas of the disks is minimized. The disk graph has input points as its vertices, and has an edge between two points if their assigned disks intersect. It is not known whether or not this problem is NP-hard. In any dimension $d\geqslant 2$ this problem is NP-hard if an upper bound on the radii of disks is given to us \cite{Chambers2011}.

\bibliographystyle{abbrv}
\bibliography{Collinear-Points.bib}

\begin{thebibliography}{1}

\bibitem{Acharyya2017a}
A.~Acharyya, M.~De, and S.~C. Nandy.
\newblock Range assignment of base-stations maximizing coverage area without
  interference.
\newblock In {\em Proceedings of the 29th Canadian Conference on Computational
  Geometry $(${CCCG}$)$}, pages 126--131, 2017.

\bibitem{Acharyya2017b}
A.~Acharyya, M.~De, S.~C. Nandy, and B.~Roy.
\newblock Range assignment of base-stations maximizing coverage area without
  interference.
\newblock {\em CoRR}, abs/1705.09346, 2017.

\bibitem{Alt2006}
H.~Alt, E.~M. Arkin, H.~Br{\"{o}}nnimann, J.~Erickson, S.~P. Fekete, C.~Knauer,
  J.~Lenchner, J.~S.~B. Mitchell, and K.~Whittlesey.
\newblock Minimum-cost coverage of point sets by disks.
\newblock In {\em Proceedings of the 22nd {ACM} Symposium on Computational
  Geometry, $($SoCG$)$}, pages 449--458, 2006.

\bibitem{Bilo2005}
V.~Bil{\`{o}}, I.~Caragiannis, C.~Kaklamanis, and P.~Kanellopoulos.
\newblock Geometric clustering to minimize the sum of cluster sizes.
\newblock In {\em Proceedings of the 13th European Symposium on Algorithms,
  $($ESA$)$}, pages 460--471, 2005.

\bibitem{Carmi2006}
P.~Carmi, M.~J. Katz, and J.~S.~B. Mitchell.
\newblock The minimum-area spanning tree problem.
\newblock {\em Computational Geometry: Theory and Applications},
  35(3):218--225, 2006.

\bibitem{Chambers2011}
E.~W. Chambers, S.~P. Fekete, H.~Hoffmann, D.~Marinakis, J.~S.~B. Mitchell,
  S.~Venkatesh, U.~Stege, and S.~Whitesides.
\newblock Connecting a set of circles with minimum sum of radii.
\newblock {\em Computational Geometry: Theory and Applications}, 68:62--76,
  2018.

\bibitem{Eppstein2016}
D.~Eppstein.
\newblock Maximizing the sum of radii of disjoint balls or disks.
\newblock In {\em Proceedings of the 28th Canadian Conference on Computational
  Geometry $(${CCCG}$)$}, pages 260--265, 2016.

\bibitem{Hsiao1992}
J.~Y. Hsiao, C.~Y. Tang, and R.~S. Chang.
\newblock An efficient algorithm for finding a maximum weight 2-independent set
  on interval graphs.
\newblock {\em Information Processing Letters}, 43(5):229--235, 1992.

\bibitem{Lev-Tov2005}
N.~Lev{-}Tov and D.~Peleg.
\newblock Polynomial time approximation schemes for base station coverage with
  minimum total radii.
\newblock {\em Computer Networks}, 47(4):489--501, 2005.

\end{thebibliography}
\end{document}